\documentclass[smallextended]{svjour3}       

\usepackage{graphicx}
\usepackage{epsfig,epstopdf,float}
\usepackage{amsmath,amssymb,amsfonts}
\usepackage{natbib}
\usepackage{color}
\usepackage{delarray}

\newcommand{\ul}{\underline}

\newcommand{\mc}{\mathcal}

\newcommand{\ds}{\displaystyle}

\newcommand {\beq} {\begin{equation}}
\newcommand {\eeq} {\end{equation}}
\newcommand {\barr} {\begin{array}}
\newcommand {\earr} {\end{array}}
\newcommand {\bear} {\begin{eqnarray}}
\newcommand {\eear} {\end{eqnarray}}
\newcommand {\bears} {\begin{eqnarray*}}
\newcommand {\eears} {\end{eqnarray*}}

\def\ind{{\mathchoice {\rm 1\mskip-4mu l} {\rm 1\mskip-4mu l}
{\rm 1\mskip-4.5mu l} {\rm 1\mskip-5mu l}}}

\journalname{Dynamic Games and Applications}
\begin{document}
\title{Stochastic Differential Games and Energy-Efficient Power Control}


\author{Fran\c cois M\'eriaux         \and
        Samson Lasaulce \and
	Hamidou Tembine
}


\institute{Fran\c cois M\'eriaux \at
              Laboratoire des Signaux et Syst\`emes \\
              \email{meriaux@lss.supelec.fr}           
           \and
           Samson Lasaulce \at
              Laboratoire des Signaux et Syst\`emes \\
              \email{lasaulce@lss.supelec.fr}
	    \and
Hamidou Tembine \at
Department of Telecommunications, \\
 \'Ecole Sup\'erieure d'\'Electricit\'e  (SUPELEC) \\
\email{tembine@ieee.org}
}

\date{Received: date / Accepted: date}

\maketitle

\begin{abstract}
One of the contributions of this work is to formulate the problem of energy-efficient power control in multiple access channels (namely, channels which comprise several transmitters and one receiver) as a stochastic differential game. The players are the transmitters who adapt their power level to the quality of their time-varying link with the receiver, their battery level, and the strategy updates of the others. The proposed model not only allows one to take into account long-term strategic interactions but also long-term energy constraints. A simple sufficient condition for the existence of a Nash equilibrium in this game is provided and shown to be verified in a typical scenario. As the uniqueness and determination of equilibria are difficult issues in general, especially when the number of players goes large, we move to two special cases: the single player case which gives us some useful insights of practical interest and allows one to make connections with the case of large number of players. The latter case is treated with a mean-field game approach for which reasonable sufficient conditions for convergence and uniqueness are provided. Remarkably, this recent approach for large system analysis shows how scalability can be dealt with in large games and only relies on the individual state information assumption.

\end{abstract}

\section{Introduction}
\label{intro}

Power control has always been recognized as an important problem for multiuser communications \citep{Foschini-tvt-1993,yates-jsac-1995}. With the appearance of new paradigms such as ad hoc networks~\citep{gupta-cdc-1997}, unlicensed band communications, and cognitive radio~\citep{fette-book-2006,mitola-1999}, the study of distributed power control policies has become especially relevant; in such networks, terminals can freely choose their power control policies and do not need to follow orders from central nodes. 
The work reported in this paper precisely falls into this framework that is, the design of distributed power control policies in multiuser networks. More precisely, the assumed network model is a multiple access channel (MAC), which, by definition, includes several transmitters and one common receiver. A brief overview of previous works about power allocation for MACs is presented by~\cite{belmega-twc-2009}. In our framework, based on a certain knowledge which includes his individual channel state information, each transmitter has to tune his power level at each time instance. The literature of power control is vast and here we will only refer to the two closest bodies of related works. In the first body, the goal is to minimize the transmit power under constraints (data rate constraints typically). While energy minimization is sought, energy-efficiency is not necessarily high when measured in terms of a benefit to cost ratio (as it is done in Physics or Economics). Clearly, energy minimization and energy-efficiency maximization are two different approaches whose relevance depends on the context under consideration (see \cite{goodman-pc-2000} and related works for more justifications) and cannot be compared in general. The results provided in this paper concern the second body of works, in which the goal is to maximize energy-efficiency which is measured as an average number of successfully decoded bits per Joule consumed by the transmitter.

In the original formulation proposed by \cite{goodman-pc-2000} and re-used in most related works \citep[e.g.,][]{meshkati-jsac-2006,bonneau-jsac-2008,lasaulce-twc-2009,buzzi-jstp-2011}, the problem of energy-efficient power control is modeled by a sequence of static games which are played independently from stage to stage. One implicit motivation behind this choice is that, in scenarios in which the channel state (i.e., the quality of the transmitter-receiver link) corresponds to i.i.d. realizations of a given random variable, correlating the power levels from block to block is a priori not relevant. But, when there exists a strategic interaction, this approach may be very suboptimal and the main drawback of the formulation of Goodman et al is precisely that it generally leads to an outcome (Nash equilibrium) which is not efficient. Motivated by this observation, \cite{LeTreustLasaulce(PowerControlRG)10} proposed a repeated game formulation of the problem. One of the strong features of their formulation w.r.t. the famous pricing approach from \cite{saraydar-com-2002} (which also aims at improving the efficiency of the game outcome on each block) is that each transmitter only needs to have individual channel state information. Although the repeated game model by \cite{LeTreustLasaulce(PowerControlRG)10} takes into account the fact that transmitters interact several/many times, the corresponding work has one major weakness: there is a need for a normalized stage game which does not depend on the realization channels. This normalization is valid only if no player has his power constraint active and even in this case, there is a loss of optimality in terms of expected utilities (note that this optimality loss is also undergone by the static game formulation with pricing proposed by \cite{saraydar-com-2002}). This is one of the main reasons why we propose a different formulation which is based on stochastic games. One of the goals of the present paper is to study the influence of long-term strategic interaction in a game with states and long-term energy constraints (e.g., the limited battery life typically) on energy-efficient power control. Indeed, in the work of \cite{goodman-pc-2000} and related references \citep[e.g.][]{meshkati-jsac-2006,bonneau-jsac-2008,lasaulce-twc-2009,buzzi-jstp-2011}, the terminals always transmit, which amounts to considering no constraint on the available energy. More specifically, the energy-efficient power control problem in MAC under long-term energy constraints is modeled by a stochastic differential game (SDG) in which the existence of a Nash equilibrium is proven. But the problem of characterizing the performance of distributed networks modeled by SDG becomes hard and even impossible when the number of players becomes large. The same statement holds for determining individually optimal control strategies. For instance, in a previous work from \cite{meriaux-dsp-2011-a}, equilibrium control strategies are proposed but they are not optimal strategies. This is where mean-field games come into play. Mean-field games \citep{lasry-jjm-2007} represent a way of approximating a stochastic differential (or difference) game, by a much more tractable model. Under the assumption of individual state information, the idea is precisely to exploit as an opportunity the fact that the number of players is large to simplify the analysis. Typically, instead of depending on the actions and states of all the players, the mean-field utility of a player only depends on his own action and state, and depends on the others through an mean-field. It seems that the most relevant work in which mean-field games have been used for power control is given by~\cite{tembine-crowncom-2010}. Compared to the latter reference, the present work is characterized by a different utility function (no linear quadratic control assumptions is made here), a more realistic channel evolution law, and the fact that the battery level of a transmitter is considered as part of a terminal state.

The remaining of the paper is organized as follows. Section~\ref{sec: oneshot game} gives a brief review of the static game formulation of the  energy-efficient power control problem. Section~\ref{sec: SDG} introduces the stochastic dynamic game modeling the energy-efficient power control game under long-term energy constraints. In Section~\ref{sec: two cases}, two particular cases of the game are studied. The single player case highlights the impact of a long-term energy constraint on the transmitter power policy. Then the large system case is modeled by a mean-field game. Section~\ref{sec: conclusion} concludes this work.

\emph{Notations:} In the following, scalars and vectors are respectively denoted by lower case symbols and underlined lower case symbols.
The vector \newline $\underline{a}_{-k} = \left(a_1, \ldots, a_{k-1},a_{k+1},\ldots, a_K\right)$ denotes the vector obtained by dropping the $k$-th component of the vector $\underline{a}$. With a slight abuse of notation, the vector $\underline{a}$ can be written $\left(a_k,\underline{a}_{-k}\right)$, in order to emphasize the influence of its $k$-th component. $\nabla_{\underline{x}} f$ and $\Delta_{\underline{x}} f$ respectively represent the gradient and the Laplacian of the function $f$ w.r.t. the vector $\underline{x}$. $div_{\underline{x}}$ is the divergence operator w.r.t. the vector $\underline{x}$. $\langle \,,\rangle_{\mathbb{E}}$ represents the scalar product in the space $\mathbb{E}$.

\section{Review of the static game formulation of the energy-efficient power control problem}\label{sec: oneshot game}

The purpose of this section is to provide a brief review of how \cite{goodman-pc-2000} formulated the power control problem. The motivation for this is twofold. It allows us to have a reference for comparison and also allows us to build in a clearer manner the SDG formulation. The communication scenario is a multiple access channel \citep{Cover-Book-91}. There are $K\geq 1$ transmitters and one receiver. Each transmitter sends a signal to a common receiver and has to choose the power level of the transmitted signal. In order to optimize his individual energy-efficiency, i.e., the ratio of his throughput to its transmit power, each transmitter has not only to adapt his power level to the quality of the channel or link between him and the receiver but also to the power levels chosen by the others. The static game formulation of this problem is as follows.

\begin{definition}[Static game model of the power control problem]\label{def:static-game}
\newline The strategic form of the static power control game is a triplet \newline $\bar{\mathcal{G}} = (\mathcal{K}, \{\mathcal{P}_i\}_{i\in\mathcal{K}}, \{\bar{u}_i\}_{i\in\mathcal{K}})$ where:
 \begin{description}
   \item[$\bullet$] $\mathcal{K} = \{1,2,...,K\}$ is the set of transmitters;
   \item[$\bullet$] $\mathcal{P}_i = [0, P_i^{\mathrm{max}}]$ is the action space of player $i\in\mc{K}$;
   \item[$\bullet$] the utility function of player $i\in\mc{K}$ is given by
   \begin{equation}
\label{eq:def-of-utility} \bar{u}_i(p_1, p_2, ..., p_K)  = \frac{R
f\left(\gamma_i(p_1,p_2, ..., p_K) \right)}{p_i} \ [\mathrm{bit} / \mathrm{J}],
\end{equation}
 where:
 \begin{itemize}
   \item
 \begin{equation}\label{eq:sinr}
\gamma_i(p_1,p_2, ..., p_K) = \frac{p_i \left|\underline{h}_i \right|^2 }{\ds{\sum_{j \in \mc{K}, j \neq i}} p_j
\left|\underline{h}_j \right|^2 +\sigma^2};
\end{equation}

   \item $\forall i \in \mc{K}$, $\underline{h}_i \in \mathbb{R}^2$ is a vector of parameters which represent the quality of the channel between transmitter $i$ and the receiver;

   \item $\sigma^2$ is a constant which models the communication noise effects at the receiver;

   \item $R$ is a constant in [bit/s] which corresponds to the communication data rate \citep[see][]{goodman-pc-2000};

   \item the function $f:\mathbb{R}^+ \rightarrow [0,1]$ is a sigmoidal or S-shaped function which represents the packet success rate; recall that a sigmoidal function is convex up to a point and then concave from this point. Additionally, $f$ is assumed to be sufficiently regular so that $\bar{u}_i$ is differentiable on $\mc{P}_i$. See the work of e.g. \cite{rodriguez-globecom-2003,meshkati-tcom-2005,belmega-tsp-2011} for a justification.
 \end{itemize}

 \end{description}
\end{definition}

\begin{figure}
\centering
\includegraphics[width=\linewidth]{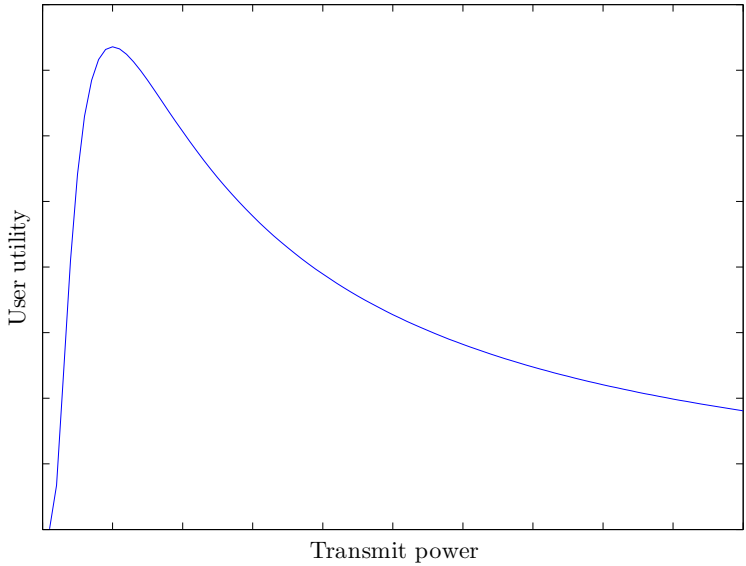}
\caption{Typical form of the utility function regarding to the transmit power.}
\label{fig: utility}
\end{figure}

The above performance metric can be seen as a tradeoff between the data rate conveyed over the air and the electromagnetic pollution (since the radiated power is concerned) in the corresponding region.
Fig.\ref{fig: utility} illustrates the typical shape of the utility function with regard to transmit power. A consequence of the use of this utility function is that the optimal transmit power generally does not make the maximum power constraint active. Precisely, this static game can be shown to be quasi-concave and has therefore a pure NE \citep[see e.g.][]{lasaulce-book-2011}. Additionally, this NE is unique and is given by:
\begin{equation}
 \forall i \in \{1,...,K \}, \
p_i^{*}= \frac{\sigma^2}{\left|\underline{h}_i \right|^2} \frac{\beta^*}{1-(K-1)\beta^{*}}
\label{eq:NE-power}
\end{equation}
where $\beta^*$ is the unique solution of the equation
\begin{equation}\label{eq: nash oneshot}
 xf'(x)-f(x)=0.
\end{equation}

It has been shown by \cite{goodman-pc-2000}, \cite{meshkati-jsac-2006}, that this equilibrium is not Pareto-optimal. Assuming that the channel coefficients are i.i.d., one could have thought that the study the one-shot energy-efficient power control problem was sufficient to understand problems of energy-efficient power control over time. But even with i.i.d. channel coefficients, there is a new phenomenon when the game is played over time which is not taken into account by the one-shot game: strategic interactions over time.

Moving to repeated game is a way to take into account these interactions and  \cite{LeTreustLasaulce(PowerControlRG)10} have shown that more efficient solutions can be obtained in such a framework. However, to account for channel variations and varying energy level in the battery, using repeated games is not sufficient. Stochastic games become necessary and are useful to further improve efficiency.
The main problem is that both the analysis and algorithm design become complicated. This is why we propose to use mean-field games. Therefore, the merit of the proposed approach is to apply stochastic games in the energy-efficient power control framework to obtain efficient solutions and simplify both the analysis and design of algorithms when the system is large. In particular, only the individual channel state information, battery state, and the interference level at the receiver are needed to implement the proposed power control policy. This is an attractive practical feature which is not available for competitive approaches based on stochastic games with finite number of players.

\section{A stochastic differential game formulation}\label{sec: SDG}

In this section, we present an SDG which is built from the static game described in the previous section. Time is assumed be continuous that is, $t\in\mathbb{R}$. This assumption has been discussed in other works on power control \citep[see e.g.][]{Foschini-tvt-1993,olama-jasp-2006,tembine-crowncom-2010}. In particular, it is relevant in scenarios in which the channels are subject to fast fading or interpreted as a limiting case for slow fading channels. From now on, we will make appear explicitly the channel states as arguments of the instantaneous utility function which will be denoted by $u_i$ instead of $\bar{u}_i$. The time horizon of the game is finite, it is the interval ranging from $T$ to $T'$. However the methodology can be extended to infinite horizon (with time average payoff) because the underlying processes are ergodic (the proposed channel model is ergodic and the remaining energy dynamics is ergodic).

\begin{definition}[SDG model of the power control problem]\label{def-sdg}
\newline The stochastic differential power control game is defined by the $5-$tuple\\ $\mc{G} = \left(\mathcal{K}, \{\mathcal{P}_i\}_{i\in\mathcal{K}}, \{\mc{X}_i\}_{i\in\mathcal{K}}, \{\mc{S}_i\}_{i\in\mathcal{K}}, \{U_i\}_{i\in\mathcal{K}}\right)$ where:
\begin{description}
  \item[$\bullet$] $\mc{K}$ and $\mc{P}_i$ are defined by Def. \ref{def:static-game};

  \item[$\bullet$] $\mc{X}_i$ is the state space of player $i$. The game state at time $t$ is defined by $\ul{X}_i(t) = [E_i(t), \underline{h}_i(t)]^{\text{T}} $ and follows the evolution law:
\begin{equation}\label{eq:channellaw}
     \mathrm{d} \underline{X}_i(t) =
     \left[
     \begin{array}{c}
     -p_i(t)\\
     \frac12\bigl(\underline{\mu} - \underline{h}_i(t)\bigr)
     \end{array}
     \right] \mathrm{d}t
     +
\left[
\begin{array}{c}
0\\
\eta
\end{array}
\right]  \text{d} \underline{\mathbb{W}}_i(t);
\end{equation}

  \item[$\bullet$] $E_i(t)$ is the energy available for player $i$ at time $t$;

 \item[$\bullet$] $\forall i \in \mathcal{K},\, \underline{\mathbb{W}}_i(t)$ are mutually independent Wiener processes of dimension $2$;

  \item[$\bullet$] $\underline{\mu}\geq 0$ and $  0 \leq \eta < +\infty$ are constants which physical interpretation is given in Prop.~\ref{prop: cdc};

    \item[$\bullet$] $\mc{S}_i$ is the set of feedback control policies for player $i$ \citep[see e.g.][]{tamer99}. A control policy will be denoted by $p_i(T\rightarrow T')$ with $p_i(t) = p_i(t,\underline{X}(t))$, and $T$, $T'$ two reals such that $T'\geq T$;

  \item[$\bullet$] the average utility function $U_i$ is defined by:
\begin{equation}
U_i\bigl(\ul{p}(T\rightarrow T')\bigr) =   \mathbb{E} \biggl[\int_T^{T'}{u_i(\underline{p}(t), \underline{X}(t)) \text{d}t} + q(\underline{X}(T')) \biggr],
\end{equation}
where $\ul{p}(T\rightarrow T')=\bigl(p_1(T\rightarrow T'), p_2(T\rightarrow T'), ..., p_K(T\rightarrow T')\bigr)$ is the control strategy profile, $\ul{X}(t) = [\ul{X}_1(t), \ul{X}_2(t),...\ul{X}_K(t)]^\mathrm{T}$ is the state profile, $q(\underline{X}(T'))$ is the utility at the final state, and $u_i$ is the instantaneous utility.
\end{description}
\end{definition}

The motivations for selecting the proposed dynamics for the state $\mc{X}_i$ are as follows. The term $ \mathrm{d} E_i(t) = -p_i(t) \mathrm{d} t$ means that the variation of energy during $\mathrm{d} t$ is proportional to the consumed power for the transmission. Indeed, the proposed model accounts for a cost when transmitting. This is fully relevant for transmitters having a finite amount of energy at disposal like cell phones, unplugged laptops, small base stations which have to be autonomous energetically speaking, etc. Such terminals have a battery with a finite amount of energy over a certain period of time and need to be recharged when empty. Our model holds over a horizon which lies strictly between two recharging instants. Over such a horizon, the available energy is a non-increasing function of the time. Although the assumed evolution law can be used for both fast and slow fading, in practice, when implementing a discrete-time version of the control policy, the energy decrease from sample to sample is stronger in the second case. Taking into account the energy of the battery in the game model changes the outcome of the game. For instance, if the battery of a transmitter is empty, the optimal power level has to be zero and cannot be the power levels recommended by the static game approach such as those given by \cite{goodman-pc-2000} or \cite{saraydar-com-2002}.

As far as time horizons are concerned, $T$ and $T'$ can be chosen arbitrarily provided that $T'\geq T$. Of course, the statistics of the evolution law are required to be stationary on this interval. In particular, $\underline{\mu}$ and $\eta$ have to be fixed. As for the order of horizon measured in second, it depends on the targeted application. Note that, in practice, if those parameters need to be known and updated, appropriate estimation schemes (with well-chosen time windows) are implemented. Let us consider two cases, fast and slow power control. For fast power control (the statistics are given by the path loss which is fixed), if the evolution approximates a discrete-time power control problem for which the (fast fading) channel is i.i.d. from block to block (in this case, $\eta$ is large), $T'-T$ is of the order of the second (a packet duration is typically 1 ms in cellular systems). In practice, most often, power levels are updated from block to block (meaning that typical updating frequency ranges from 100 Hz to a few kHz). As shown by \cite{Foschini-tvt-1993}, deriving continuous-time power control policies are still useful since it is possible to build practical discrete-time algorithms from them. Additionally, when one studies the convergence of these discrete time algorithms, often, it amounts to studying continuous-time dynamics. For slow power control, if the evolution law represents the variation of the pathloss/shadowing/slow fading (as assumed by \cite{olama-jasp-2006} and related works), the time horizon is much larger, typically of the order of a minute or more (this depends on the mobile velocity of course). For both fast and slow power control, the number of samples to approximate the integrals can be of the same order since the updating frequencies are different.

Although being simple, the dynamics for the channel gain $\ul{h}_i(t)$ capture several typical effects in wireless communications. Before commenting on these effects, let us state a property of the random process $\ul{h}_i(t)$.
\begin{proposition}[Channel dynamics property] \label{prop: cdc}
\newline Let $\ul{h}_i(t) = (x_i(t), y_i(t))$ be governed by the dynamics defined in the SDG $\mc{G}$, then we have that:
\begin{equation}
\begin{array}{ll}
 \ds{\lim_{t\to + \infty}} \mathbb{E}[\underline{h}_i(t)] = \ul{\mu}, \\
\ds{\lim_{t\to + \infty}} \mathbb{E}[|\underline{h}_i(t)|^2] - \mathbb{E}[|\underline{h}_i(t)|]^2 = 2\eta^2.
\end{array}
\end{equation}
The stationary probability density functions $m_x:\; \mathbb{R} \to \mathbb{P}(\mathbb{R})$, $m_y:\; \mathbb{R} \to \mathbb{P}(\mathbb{R})$ of the two components $x_i$, $y_i$ of $h_i$ are:
\begin{equation}
\left\{
\begin{array}{cc}
 m_x(x_i) = \frac{1}{\eta \sqrt{2 \pi}} e^{-\frac{(x_i - \mu_x)^2}{2 \eta^2}}, \\
 m_y(y_i) = \frac{1}{\eta \sqrt{2 \pi}} e^{-\frac{(y_i - \mu_y)^2}{2 \eta^2}}.
\end{array}
\right.
\end{equation}
\end{proposition}
The proof of this result is simple and provided in Appendix \ref{ap:proof1}. This shows that the proposed dynamics allow one to model Rician channels namely, channels with zero-mean gains; this is possible by tuning $| \ul{\mu} | \geq 0$ which represents the Rice component. Also, by choosing the variance $2\eta^2$ in an appropriate manner, one can account for the fading effects. As a relevant comment, note that the assumed channel dynamics can also be seen as a limiting case of the important Gauss-Markov discrete-time model used to model time correlation for the channel gains \citep{agarwal-tit-2012}. To conclude on the choice of these dynamics, we will see in Sec. \ref{sec: two cases} that they also possess an interesting property for the mean-field dynamics under investigation.

At this point, we can define a Nash equilibrium of the SDG $\mc{G}$ and state our existence result.

\begin{definition}[Nash equilibrium of $\mc{G}$]
\newline A control strategy profile $\underline{p}^*(t,\underline{X}(t))=(p^*_1(t,\underline{X}(t)),
\ldots,p^*_K(t,\underline{X}(t))$
is a feedback Nash equilibrium of the SDG if and only if $\forall i \in \mathcal{K}$, $p_i^*$ is a solution of the control problem
\begin{equation}
\sup_{p_i(T \to T')}\mathbb{E} \biggl[\int_T^{T'}{u_i\biggl(p_i(t,\underline{X}(t)),
\underline{p}^*_{-i}(t,\underline{X}(t)), \underline{X}(t)  \biggr) \text{d}t} + q(\underline{X}(T')) \biggr],
\end{equation}
subject to
\begin{equation}
     \mathrm{d} \underline{X}(t) =
     \left[
     \begin{array}{c}
     -[p_i(t),\ul{p}^*_{-i}(t)]^{\text{T}}\\
     \frac12\bigl( \ul{\ind} \otimes \underline{\mu} - \underline{h}(t)\bigr)
     \end{array}
     \right] \mathrm{d}t
     +
\left[
\begin{array}{c}
0\\
\eta
\end{array}
\right]  \text{d} \underline{\mathbb{W}}(t),
\end{equation}
where $\ul{\ind} = (1, 1, ..., 1) \in \mathbb{R}^K$, $\ul{h}(t) = [\ul{h}_1(t), \ul{h}_2(t), ...\ul{h}_K(t) ]^\mathrm{T}$, $\otimes$ stands for the Kronecker product, and $\underline{\mathbb{W}}(t) = [\ul{\mathbb{W}}_1(t), \ul{\mathbb{W}}_2(t), ..., \ul{\mathbb{W}}_K(t)]^{\text{T}}$.
\end{definition}

Regarding to the existence of a Nash equilibrium, one can state the following proposition (the function $f$ is defined by Def.~\ref{def:static-game}).
\begin{proposition}[Existence of a Nash equilibrium in $\mc{G}$]
\newline A sufficient condition for the existence of a Nash Equilibrium in $\mc{G}$ is that for all $(\theta_0,\gamma_0)$ such that $f'(\gamma_0)\gamma_0-f(\gamma_0) = \theta_0\gamma_0^2$, we have that $2\theta_0 - f''(\gamma_0) \neq 0$.
\end{proposition}
For the proof of this result and for clarifying some points in the sequel, we will use the (auxiliary) Bellman function, which is defined by:
\begin{equation}
v_i(T,\underline{X}(T))=\sup_{p_i(T \to T')}\mathbb{E} \biggl[\int_T^{T'}{u_i(\underline{p}(t), \underline{X}(t)) \text{d}t} + q(\underline{X}(T')) \biggr].
\end{equation}

\begin{proof}
According to~\cite{Bressan10noncooperativedifferential}, a sufficient condition for the existence of a Nash equilibrium for the SDG is the existence of a solution to the Hamilton-Jacobi-Bellman-Fleming~\citep{fleming1993controlled} equation for each transmitter

\begin{equation}
\label{eq: HJBF discrete}
\begin{aligned}
0 = &\sup_{p_i(T \to T')}\biggl[ u_i(\underline{X}(t),\underline{p}(t,\underline{X}(t))) - p_i(t,\underline{X}(t)) \frac{\partial v_i(t,\underline{X}(t))}{\partial E_i}\biggr] \\
+ &\frac12 \langle \underline{\mu} - \underline{h}_i(t),\nabla_{\underline{h}_i} v_i(t,\underline{X}(t))\rangle_{\mathbb{R}^2}  +\frac{\partial v_i(t,\underline{X}(t))}{\partial t} +\frac{\eta^2}{2} \Delta_{\underline{h}_i} v_i(t,\underline{X}(t)).
\end{aligned}
\end{equation}

\noindent
There exists a solution if the function
\begin{equation}
\begin{aligned}
H\biggl(\underline{X}(t),\underline{p}_{-i}&(t,\underline{X}(t)),\frac{\partial v_i(t,\underline{X}(t))}{\partial E_i}\biggr)= \\
&\sup_{p_i(T \to T')}\biggl[ u_i(\underline{X}(t),\underline{p}(t,X(t))) - p_i(t,\underline{X}(t)) \frac{\partial v_i(t,\underline{X}(t))}{\partial E_i}\biggr]
\end{aligned}
\end{equation}

\noindent
is smooth (see~\cite{evans2010partial} for more details). And with a similar reasoning as~\cite{meriaux-dsp-2011-b}, we can show that finding optimal power control
\begin{equation}
\underline{p}^*(T \to T') = (p_1^*(T \to T'),\ldots,p_n^*(T \to T'))
\end{equation}

\noindent
such that $\forall i \in \mathcal{K}$, $p_i^*(T \to T') \in $
\begin{equation}
\arg \max_{p_i(T \to T')}\biggl[ u_i\biggl(\underline{X}(t),p_i(t,\underline{X}(t)),\underline{p}^*_{-i}(t,\underline{X}(t))\biggr) - p_i(t,X(t)) \frac{\partial v_i(t,\underline{X}(t))}{\partial E_i}\biggr]
\end{equation}

\noindent
amounts to solving $\forall i \in \mathcal{K}$, $\forall t \in [T,T']$
\begin{equation}
\label{eq : condimax}
f'(\gamma_i(t))\gamma_i(t)-f(\gamma_i(t)) = \frac{\gamma_i(t)^2}{R}\frac{\partial v_i(t,\underline{X}(t))}{\partial E_i}\biggl(\frac{\sigma^2 +\sum_{j \neq i}^n |\underline{h}_j(t)|^2 p_j^*(t)}{|\underline{h}_i(t)|^2}\biggr)^2.
\end{equation}

\noindent
Note that we consider that $\frac{\partial v_i(t,\underline{X}(t))}{\partial E_i} \geq 0$, otherwise the optimal power $p_i^*(t) \to \infty$.
The existence of a non-zero solution depends on the term
\begin{equation}
\theta_i = \frac{1}{R}\frac{\partial v_i(t,\underline{X}(t))}{\partial E_i}\biggl(\frac{\sigma^2 +\sum_{j \neq i}^n |\underline{h}_j(t)|^2p_j^*(t)}{|\underline{h}_i(t)|^2}\biggr)^2.
\end{equation}

\noindent
It can be checked that there exists a threshold $\theta_{\max}$ such that if $\theta_i < \theta_{\max}$, there exists a unique global maximizer $\gamma^*$ different from $0$ and if $\theta_i \geq \theta_{\max}$, $0$ is the global maximizer.

\noindent
We call for the implicit function theorem to state smoothness of
\newline
$H\biggl(\underline{X}(t),\underline{p}_{-i}(t,\underline{X}(t)),\frac{\partial v_i(t,\underline{X}(t))}{\partial E_i}\biggr)$, we define
\begin{equation}
\begin{aligned}
g:\;&[0,\theta_{\max}[\times\mathbb{R}^+ \to \mathbb{R} \\
&(\theta,\gamma) \to f'(\gamma)\gamma-f(\gamma) - \theta\gamma^2.
\end{aligned}
\end{equation}

\noindent
$g$ is $C^{\infty}$, then if $g(\theta_0,\gamma_0)=0$, there exists $\varphi: \mathbb{R} \to \mathbb{R}$ such that $\gamma_0 = \varphi(\theta_0)$. \newline $\varphi$ is $C^{\infty}$  and
\begin{equation}
\frac{\partial \varphi}{\partial \theta}(\theta_0) = -\frac{\frac{\partial g(\theta_0,\gamma_0)}{\partial \theta}}{ \frac{\partial g(\theta_0,\gamma_0)}{\partial \gamma}}.
\end{equation}

\noindent
In our case, it writes
\begin{equation}
 \frac{\partial \varphi}{\partial \theta}(\theta_0) = \frac{\gamma_0}{2\theta_0 - f''(\gamma_0)}.
\end{equation}
If $2\theta_0 - f''(\gamma_0) \neq 0$, then smoothness is ensured. \hfill \ensuremath{\blacksquare}
\end{proof}
Remarkably, the proposed sufficient condition holds for all particular choices of efficiency function made in the literature. In particular, it holds for the information-theoretic choice of \cite{belmega-tsp-2011}. Indeed, if $f(x) = e^{-\frac{a}{x}}$, $a\geq 0$ we have that:
\begin{equation}
 \left\{
\begin{array}{ll}
e^{-\frac{a}{\gamma_0}}(\frac{a}{\gamma_0} - 1) = \theta_0 \gamma_0^2, \\
2\theta_0 - (\frac{a^2}{\gamma_0^4}- \frac{2a}{\gamma_0^3})e^{-\frac{a}{\gamma_0}} \neq 0,
\end{array}
\right.
\end{equation}
which gives
\begin{equation}
 2\theta_0 - \biggl(\frac{a}{\gamma_0}- 2\biggr)\frac{a\theta_0}{a-\gamma_0} \neq 0.
\end{equation}

While Nash equilibrium uniqueness is an attractive property of the static game $\bar{\mc{G}}$, this property is not easy to be verified for the SDG $\mc{G}$. Rather, this type of games has generally a large number of equilibria. Concerning the explicit determination of Nash equilibrium power control policies, it has to be mentioned that this task is also not easy a priori. Precisely, as written in (\ref{eq: HJBF discrete}), the determination of a Nash equilibrium requires to solve a system of $K$ Hamilton-Jacobi-Bellman-Fleming equations, coupled by the state $\underline{X}(t)$. Interestingly, there are two limiting cases of $\mc{G}$ for which both uniqueness and existence issues are much easier. The first special case is when there is only one player (note that the corresponding optimization problem has not been studied in the literature). The second case is when the number of players is large, making the mean-field game analysis fully relevant. These two cases are the purpose of the next section.

\section{Two relevant special cases of $\mc{G}$: $K= 1$ and $K \rightarrow +\infty$}\label{sec: two cases}

\subsection{The single-player case ($K= 1$)}

One of the interests in analyzing  the single-user case is to separate the effects due to the long-term energy constraint from those due to interaction between players (two effects can incite the transmitter to be off, namely a bad channel state, and high interference level). Indeed, in the single-player case only the former effects appear. Below, it is proven through simple equations that the transmitter is not always on. The fraction of time during which the transmitter has to be off is approximated, which is of practical interest. When moving to the case of several players, players can also have an interest in not transmitting (as observed by~\cite{meriaux-dsp-2011-b}), making appropriate time-sharing policies natural equilibria. Summing up, both reducing the interference between the transmitters and long-term energy constraint can incite a transmitter to be off.

In this section, we therefore study the special case of the game $\mc{G}$ in which there is only one player. In this context, there is obviously no interaction between players and the main interest of this case is to show the influence of the long-term energy constraint.
From the preceding section, it can be seen that determining an optimal control policy amounts to solving the following equation in $\gamma_1(t)$:
\begin{equation}\label{eq:gamma-star}
 f'(\gamma_1(t))\gamma_1(t)-f(\gamma_1(t)) = \frac{\gamma_1(t)^2}{R} \frac{\partial v_1(t,\underline{X}_1(t))}{\partial E_1} \frac{\sigma^4}{|\underline{h}_1(t)|^4}.
\end{equation}
By denoting $\gamma_1^*(t)$ the largest solution of the above equation, an optimal control policy follows (see (\ref{eq:sinr})):
\begin{equation}
\label{eq:pc-scheme}
p_1^*(t) = \frac{\sigma^2}{| \underline{h}_1(t) |^2} \gamma_1^*(t)),
\end{equation}
with $\gamma_1^*(t) \geq 0 $. Interestingly, $\gamma_1^*(t) = 0$ can also occur at given time instants. In particular, this depends on the channel quality $|\underline{h}_1(t)|$. If the latter is too low, transmitting is not energy-efficient, leading to a vanishing transmit power. With the same reasoning as \cite{meriaux-dsp-2011-b}, the fraction of the time during which the transmitter is off can be assessed by using a simple lower bound on the probability that $\gamma_1^*(t) = 0$:
\begin{equation}
\mathrm{Pr}\left[\max f'' \leq 2 \frac{\partial v_1(t,\underline{X}_1(t))}{\partial E_1}\frac{\sigma^4 }{R |\underline{h}_1(t)|^4}\right] \leq \mathrm{Pr}[\gamma_1^*(t) = 0].
\end{equation}

Many Monte Carlo simulations have shown that this lower bound is reasonably tight, one of them is provided in figure \ref{fig: compa proba}. This result is of practical interest since it allows one to quantify the impact of a long-term energy constraint on power control policies, which is one of the goals of this paper.

\begin{figure}
\centering
\includegraphics[width=\linewidth]{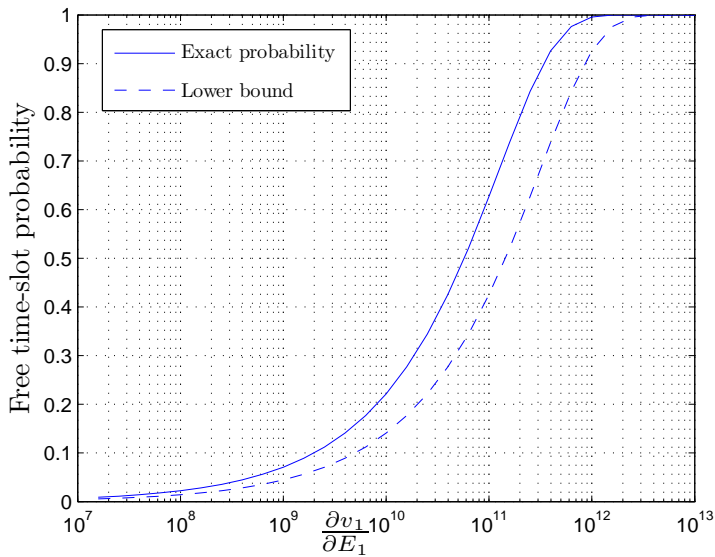}
\caption{Exact probability of the transmitter being off and lower bound of this probability depending on $\frac{\partial v_1(t,\underline{X}_1(t))}{\partial E_1}$.}
\label{fig: compa proba}
\end{figure}

\subsection{The mean-field game analysis ($K \rightarrow +\infty$)}
\subsubsection{Convergence to the mean-field game}

We have mentioned in the preceding sections that the SDG becomes more and more difficult to analyze when the number of players increases. However, our problem has a special structure which can be exploited and simplifies the problem when $K$ is large. Indeed, from a given player standpoint, what matters in terms of utility is a weighted sum of the played actions. The relevant quantity which affects the utility of player $i$ is the following quantity:
\begin{equation}
I_i(t) = \frac{1}{K} \sum_{j \in \mc{K},  j \neq i} p_j(t)|\underline{h}_j(t)|^2.
\end{equation}
The quantity $I_i(t)$ is called the interference in communication networks. Here, we have normalized this quantity.
There are many justifications of practical interest for this normalization \citep[see e.g.][]{tulino-book-04}. For example, it is fully justified in random CDMA systems \citep{meshkati-jsac-2006}. If it can be proven that if $I_i(t)$ converges, then the game $\mathcal{G}$ converges to a mean-field game. To justify the convergence here, the conventional weak law of large numbers is not applicable since the random processes $p_j(t)|\underline{h}_j(t)|^2$ are interdependent. However, there still exist some conditions under which convergence is ensured. One of them is the exchangeability or indistinguishability property which is defined next.
\begin{definition}[Exchangeability]
\newline The states $X_1, X_2, ..., X_K$ are said to be exchangeable in law under the feedback strategies $\alpha$ if they generate a joint law which is invariant by permuting the players indices , i.e.,
\begin{equation}
\forall\ K,\ \mathcal{L}\left( \underline{X}_1,\ldots,\underline{X}_K\ | \ \alpha \right)=\mathcal{L}\left( \underline{X}_{\pi(1)},\ldots, \underline{X}_{\pi(K)}\ | \  \alpha\right),
\end{equation} for any bijection $\pi$ (one-to-one mapping) defined over $\{1,\ldots,K\}.$
\end{definition}
To guarantee this property for the game under study we make the following assumptions:

\begin{itemize}
  \item each player only knows his individual state;

   \item each player implements an homogeneous admissible control:
\begin{equation}
p_i(t) = \alpha(t,\underline{X}_i(t));
\end{equation}

  \item $\mathbb{E}\biggl[\ds{\int_T^{T'}}\alpha(t,\underline{X}_i(t))^2\ \text{d}t\biggr] < +\infty$.
\end{itemize}
As a consequence of the exchangeability property in $\mathcal{G}$, the game now comprises generic players.
From now on, we call $\underline{s}(t)=[E(t),\underline{h}(t)]^{\text{T}}$ the generic individual state of a player. The state dynamics of a generic player is given by the following stochastic differential equations (SDEs):
\begin{equation}\label{eq: state_evo}
\left\{
\begin{array}{ll}
\text{d} E(t) = - \alpha(t,\underline{s}(t)) \text{d} t, \\
\text{d} \underline{h}(t) = \frac12\bigl(\underline{\mu} - \underline{h}(t)\bigr)\text{d}t + \eta \text{d} \underline{\mathbb{W}}(t).
\end{array}
\right.
\end{equation} where $\underline{\mu},\eta$ are time-independent.

Before discussing the convergence of $\mathcal{G}$ to a mean-field game, we introduce some notations to define the mean-field game concept properly.
Let $(\Omega,\mathbb{P},\mathcal{F})$
be a complete filtered probability space, on which a one-dimensional standard Brownian
motion $\mathcal{W}$ is defined with $\mathbb{F}=(\mathcal{F}_t)_{t\geq 0},\ $  being its natural filtration augmented by all the $\mathbb{P}-$null sets (sets of measure-zero with the respect $\mathbb{P}.$) The filtration $\mathcal{F}_t$ will be combined with the one generated by the initial state of the players.
Note that at time $t$ the trajectory generated by state dynamics (\ref{eq: state_evo}) is in $\mathcal{F}_t.$ The state given by the channel gains and the battery levels follows a certain distribution which evolves over time: this distribution is called the mean field.

\begin{proposition}[Convergence to the mean-field game]
 \newline If the states $(\underline{X}_i(t))_{i,t}$ and the admissible controls $(\alpha(t,\underline{X}_i(t)))_{i,t}$ preserve the indistinguishability property and the $(\underline{X}_i(0))_i$ are indistinguishable, then the stochastic differential game converges to a mean-field game.
\end{proposition}

\begin{proof}
\begin{equation}
\begin{aligned}
I_i(t) = & \frac{1}{K} \biggl[ \sum_{j=1}^K \alpha(t,\underline{X}_j(t))|\underline{h}_j(t)|^2 - \alpha(t,\underline{X}_i(t))|\underline{h}_i(t)|^2 \biggr], \\
= & \int_{\underline{s}} |\underline{h}|^2 \alpha(t,\underline{s}) M_t^K(\text{d}\underline{s}) -\frac{\alpha(t,\underline{X}_i(t))|\underline{h}_i(t)|^2}{K}, \\
\end{aligned}
\end{equation}
with
\begin{equation}
M_t^K = \frac{1}{K} \sum_{j=1}^K \delta_{\underline{s}_j(t)},
\end{equation}
where $ \delta_{\underline{s}_j(t)}$ is the Dirac measure concentrated at $\underline{s}_j(t)$.

If the number of transmitters becomes very large ($K \to \infty$), we can consider that we have a continuum of transmitters. The convergence of the interference term when $K \to \infty$ needs to be proven. Using admissible control, $\mathbb{E} [\alpha(t,\underline{X}_i(t))|\underline{h}_i(t)|^2] <  \infty$, then
\begin{equation}
\lim_{K \to \infty}\frac{\alpha(t,\underline{X}_i(t))|\underline{h}_i(t)|^2}{K} = 0.
\end{equation}
To prove $I_i(t)$ converges weakly, it suffices to prove $\int_{\underline{s}} |\underline{h}|^2 \alpha(t,\underline{s}) M_t^K(\text{d}\underline{s})$ converges weakly. A sufficient condition is the weak convergence of the process $M_t^K$. Since we have chosen  the control law that preserves the indistinguishability property, we can use the work by \cite{cdc2011} which states that there exists a distribution $m_t$ such that
\begin{equation}
m_t = \lim_{K \to \infty}M_t^K.
\end{equation}
\end{proof}
This distribution $m_t$ is the mean-field.
In our case, the evolution of the state of the transmitters does not depend on the index of the transmitters and each each transmitter state law satisfies the system  (\ref{eq: state_evo}).
Thus, the indistinguishability property holds. \hfill \ensuremath{\blacksquare}

\subsubsection{Solution to the mean-field response problem}
We are now in position to define the solution concept of the SDE (\ref{eq: state_evo}).
\begin{definition} \label{defifokker}
Let $T,T'>0$ such that $[T, T']$ is the horizon of the game.
We say that the state distribution $m_t(\underline{s})$ is a weak solution to the state dynamics (\ref{eq: state_evo}) if $m_t$ is integrable over $[T,T']$ and for any infinitely continuously differentiable function $\phi$  over $ (T,T')\times \mathbb{R}^3$ with compact support (test function), one has
\begin{equation}
\mathbb{E}_{m_T}[\phi_{T}(.)]+\int_{T}^{T'} \mathbb{E}_{m_t}\left[\partial_t\phi_t(.)-\langle \nabla_{\underline{s}}\phi_t, D^*(t,\underline{s})\rangle_{\mathbb{R}^2}+\frac{\eta^2}{2}\Delta_{\underline{s}}\phi_t  \right]=0,
\end{equation}
where
$\mathbb{E}_{m_t}$ is the expectation with the respect to $m_t$ and $D^*(t,\underline{s})$ is the drift vector $[-\alpha(t,\underline{s}),\frac{1}{2}(\underline{\mu}-\underline{h}(t))]^{\text{T}}.$
\end{definition}
Assuming that $D^*$ is sufficiently regular in time $t$ and in state $\underline{s},$ we examine the existence of
 a  solution
$\underline{s}(.)$ which is $\mathcal{F}$-adapted,
\begin{equation}
 \mathbb{E}\left[\sup_{t\in [T,T']}\ | \underline{s}(t)|^2 \right] <+\infty,
\end{equation}
 and $\underline{s}(�)$ has continuous paths.

A consequence of It\^o's formula \citep[see e.g.][]{Karatzas_Shreve_1991} states that the law of the SDE starting from distribution $m_T$ is a weak solution of the partial differential equation  (\ref{eq: state_evo}).
We deduce, from the Definition \ref{defifokker}, the equation satisfied by the distribution of the states, i.e.,
 {\it Fokker-Planck-Kolmogorov forward equation} is given by:
\begin{equation}
\frac{\partial m_t}{\partial t} - \frac{\partial}{\partial E}(m_t \alpha )+div_{\underline{h}}(m_t  \frac12(\underline{\mu} - \underline{h}) ) - \frac{\eta^2}{2} \Delta_{\underline{h}} m_t=0.
\end{equation}
With the new parameters of the game defined as:
\begin{equation}
\hat{I}(t,m_t) = \int_{\underline{s}} |\underline{h}|^2 \alpha(t,\underline{s}) m_t(\text{d}\underline{s}),
\end{equation}
\begin{equation}
\widehat{\gamma}(\underline{s}(t),m_t) = \frac{p(t)|\underline{h}(t)|^2}{\sigma^2+ \hat{I}(t,m_t)},
\end{equation}
\begin{equation}
\hat{u}(t) = \frac{R f(\widehat{\gamma}(\underline{s}(t),m_t))}{p(t)} =: \hat{r}(\underline{s}(t),p(t),m_t),
\end{equation}
we can formulate the mean-field response problem in which each generic user best-responds to the mean-field:
\begin{equation}
\hat{v}_T = \sup_{p(T \to T')} \mathbb{E} \biggl[ q(\underline{s}(T')) + \int_T^{T'} \hat{r}(\underline{s}(t),p(t),m_t^*) \text{d}t \biggr],
\end{equation}
where $m_t^*$ is the mean-field optimal trajectory and $m_T$ is assumed to be given.
A solution of the mean-field response problem is a solution of
\begin{equation}
\label{eq:HJBF-FPK}
\left\{
\begin{array}{ll}
\frac{\partial \hat{v}_t}{\partial t} + \tilde{H}(\underline{s}(t),\frac{\partial \hat{v}_t}{\partial E},m_t) + \frac12 \langle \underline{\mu} - \underline{h},\nabla_{\underline{h}} \hat{v}_t\rangle_{\mathbb{R}^2}  +\frac{\eta^2}{2} \Delta_{\underline{h}} \hat{v}_t=0, \\
\frac{\partial m_t}{\partial t} + \frac{\partial}{\partial E}(m_t \frac{\partial}{\partial u'} \tilde{H}(\underline{s}(t),\frac{\partial \hat{v}_t}{\partial E} ,m_t) )+div_{\underline{h}}(m_t  \frac12(\underline{\mu} - \underline{h}) ) = \frac{\eta^2}{2} \Delta_{\underline{h}} m_t, \\
\end{array}
\right.
\end{equation}
with $\hat{v}_{T'} = q(s(T')),\ $  $m_T$ known and
\begin{equation}
\tilde{H}(\underline{s},u',m) = \sup_p\{ \hat{r}(\underline{s},p,m) -p.u'\}.
\end{equation}

\noindent
As for the SDG case, the first equation is a Hamilton-Jacobi-Bellman-Fleming equation. But it is now coupled with a Fokker-Planck-Kolmogorov equation. The former one is a backward equation whereas the latter one is a forward equation. The other main difference between the SDG and the MFG is that in the former, each transmitter needs full knowledge of the channel states and the transmit powers of the other players to compute the outcome of the game, whereas in the latter only the knowledge of individual state and the mean-field is required to compute the outcome. In our model, although this mean-field cannot be directly known by the transmitters or the common receiver, it aggregates in the interference term which can be known by the receiver. Consequently, once the game is solved, given the interference at the receiver (which can be broadcast to every transmitters), its channel state and its battery state, each transmitter knows the power value it should use. The solution provides every transmitter a function of interference, channel state and battery state which output is the stable power control policy.

\subsubsection{Uniqueness of the solution}
Interestingly, a sufficient condition can be given for the solution of the mean-field response problem to be unique. First, we recall the definition of positiveness for an operator.

\begin{definition}[Positiveness of an operator]
\newline We say that the operator $\mathcal{O}:\, \mathcal{E} \to \mathcal{E}$ is positive (denoted by $\mathcal{O} \succ 0$) if
\[
 \forall x \neq 0_{\mathcal{E}},\; \langle x, \mathcal{O} x \rangle_{\mathcal{E}} > 0,
\]
where $0_{\mathcal{E}}$ is the neutral element of $\mathcal{E}$.
\end{definition}

\begin{proposition}[Uniqueness of the mean-field response problem solution]
A sufficient condition for the uniqueness of the solution to the mean-field response problem is for all triplet $(\underline{s},u',m) \in \mathbb{R}^3 \times \mathbb{R} \times \mathbb{P}(\mathbb{R}^3)$,
\begin{equation}
  \begin{aligned}
  -&\frac{1}{m}\frac{\partial}{\partial m} \tilde{H}(\underline{s},u',m) \succ 0,\\
&\frac{\partial^2}{\partial u'^2} \tilde{H}(\underline{s},u',m) > 0, \\
& \frac{\partial^2}{\partial m \partial u'} \tilde{H}(\underline{s},u',m) -\frac12  \frac{(\frac{\partial^2}{\partial m \partial u'} \tilde{H}(\underline{s},u',m))^2}{\frac{\partial^2}{\partial u'^2} \tilde{H}(\underline{s},u',m)} \succ 0.
 \end{aligned}
\end{equation}
\end{proposition}
The proof of this proposition is given in Appendix~\ref{ap:proof2}.

\section{Conclusion}\label{sec: conclusion}

This paper provides a stochastic differential game formulation of the energy-efficient power control problem initially introduced in \cite{goodman-pc-2000}. This formulation allows one to better optimize the global efficiency of the network, account for long-term energy constraints, and take into account propagation effects such as time correlation for the channel gains. The problem is that this model becomes intractable when the number of transmitters becomes large. Instead of seeing large networks as a curse, they can be seen as a blessing since under the large system assumption, the game can be approximated by a mean-field game. Under the assumption of individual state information, the idea is precisely to exploit the large number of players to simplify the analysis.
The authors believe the present paper provides several interesting results to go into this direction but must admit that the numerical analysis and the design aspect still require a lot of efforts to make this approach more practical (in the same way as random matrix theory was initially introduced in the wireless literature by \cite{tse-tit-1999} and shown to be of practical interest later on by e.g., \cite{dumont-it-2010}). Nonetheless, this paper provides several interesting results on the mean-field game approach. Under the exchangeability assumption, the stochastic differential game is shown to converge to a mean-field game as the number of players increases. This new game simplifies and even makes possible the equilibrium analysis since the equilibrium derivation only requires the knowledge of the individual state and the mean-field to solve a system of two equations. To be more precise, each transmitter needs to know its channel state, its battery state and the instantaneous interference it undergoes. 
For this, the receiver only needs to feed back the instantaneous interference. In the mean-field model, this instantaneous interference is the same for all the transmitters. Hence, the required signal is a broadcast to all the transmitters.
Remarkably, this signal is fully scalable since, in theory (up to quantization effects), the amount of signalling does not depend of the number of transmitters. This framework allows us to derive simple sufficient conditions for the existence and uniqueness of an equilibrium power control policy.

\appendix

\section{Proof of Proposition 1}\label{ap:proof1}

 Both results can be proven by using Ito's formula \citep[see e.g.][]{Karatzas_Shreve_1991}. For the mean, from (\ref{eq:channellaw}), one has
\begin{equation}
 \text{d} \mathbb{E}[\underline{h}_i(t)] = \frac12 (\underline{\mu} - \mathbb{E}[\underline{h}_i(t)]) \text{d}t,
\end{equation}
then $\mathbb{E}[\underline{h}_i(t)] = \underline{\mu}(1 - e^{-\frac t2}) + \underline{h}_i(0)e^{-\frac t2}$. The limit when $t$ goes to $+\infty$ writes
\begin{equation}
 \lim_{t\to + \infty} \mathbb{E}[\underline{h}_i] = \underline{\mu}.
\end{equation}

For the variance, assume that for the two components $x_i$ and $y_i$ of $h_i$:
\begin{equation}
\begin{array}{ll}
 \text{d} x_i(t) = g_x(t) \text{d}t + \eta \text{d}\mathbb{W}_x, \\
 \text{d} y_i(t) = g_y(t) \text{d}t + \eta \text{d}\mathbb{W}_y,
\end{array}
\end{equation}
with $\mathbb{W}_x$ and $\mathbb{W}_y$ two independent Wiener processes of dimension $1$.
Then
\begin{equation}
 \text{d} x_i(t)^2 = (2x_i(t) g_x(t) + \eta^2) \text{d}t + 2 x_i(t) \eta \text{d}\mathbb{W}_x,
\end{equation}
and
\begin{equation}
 \text{d} \mathbb{E}[x_i(t)^2] = \mathbb{E}[(2x_i(t) g_x(t) + \eta^2)] \text{d}t.
\end{equation}

\noindent
If $g_x(t) = 0$, $\mathbb{E}[x_i(t)^2] = x_i(0)^2 + \eta^2 t$ and $\lim_{t \to \infty} \mathbb{E}[x_i(t)^2] = \infty$. That is the reason why a deterministic term is needed in (\ref{eq:channellaw}). With $g_x(t) = \frac 12(\mu_x - x_i(t))$, one has
\begin{equation}
 \text{d}\mathbb{E}[x_i(t)^2] = \mathbb{E} [2 x_i(t)\frac 12(\mu_x - x_i(t)) + \eta^2] \text{d}t,
\end{equation}
then
\begin{equation}
 \frac{\text{d}\mathbb{E}[x_i(t)^2]}{\text{d}t} = -\mathbb{E} [x_i(t)^2] + \mu_x \mathbb{E} [x_i(t)] + \eta^2.
\end{equation}

\noindent
The solution of this differential equation has the form
\begin{equation}
\begin{aligned}
 \mathbb{E}[x_i(t)^2] = &\biggl(x_i(0)^2 + \int_0^t\bigl(\mu_x\mathbb{E}[x_i(t')] + \eta^2\bigr)e^{ t'} \text{d}t'\biggr)e^{- t}, \\
= & \biggl[\bigl(\mu_x^2 + \eta^2\bigr)e^{t'} + 2 \mu_x \bigl(x_i(0) - \mu_x\bigr)e^{\frac{t'}{2}}\biggr]_0^t e^{- t}, \\
= & \bigl(\mu_x^2 + \eta^2\bigr)(1-e^{-t}) + 2 \mu_x \bigl(x_i(0) - \mu_x\bigr)(e^{-\frac{t}{2}} -e^{-t}),
\end{aligned}
\end{equation}
thus
\begin{equation}
 \lim_{t\to + \infty} \mathbb{E}[x_i(t)^2] = \mu_x^2 + \eta^2.
\end{equation}

\noindent
The analogous is true for $y_i$. Hence we have
\begin{equation}
 \lim_{t\to + \infty} \mathbb{E}[|\underline{h}_i(t)|^2] = |\underline{\mu}|^2 + 2\eta^2,
\end{equation}
and finally
\begin{equation}
 \lim_{t\to + \infty} \mathbb{E}[|\underline{h}_i(t)|^2] - \mathbb{E}[|\underline{h}_i(t)|]^2 = 2\eta^2.
\end{equation}

Regarding to the probability density functions, applying the Kolmogorov forward equation to the state $\underline{h}_i$ with the dynamics given in (\ref{eq:channellaw}), one has for the component $x_i$
\begin{equation}\label{eq:KFh}
\begin{aligned}
 \frac{\partial m_x(x_i,t)}{\partial t} = &- \frac{\partial}{\partial x_i} \bigl[m_x(x_i,t)\frac12 (\mu_x - x_i)\bigr] + \frac{\eta^2}{2} \frac{\partial^2 m_x(x_i,t)}{\partial x_i^2}, \\
= &\frac12 m_x(x_i,t) - \frac12 (\mu_x - x_i) \frac{\partial m_x(x_i,t)}{\partial x_i} + \frac{\eta^2}{2} \frac{\partial^2 m_x(x_i,t)}{\partial x_i^2}.
\end{aligned}
\end{equation}
The stationary case gives
\begin{equation}\label{eq:KFsta}
0= \frac12 m_x(x_i) - \frac12 (\mu_x - x_i) \frac{\partial m_x(x_i)}{\partial x_i} + \frac{\eta^2}{2} \frac{\partial^2 m_x(x_i)}{\partial x_i^2}.
\end{equation}
One can check that $\hat{m}_x(x_i) = \frac{1}{\eta \sqrt{2 \pi}} e^{-\frac{(x_i - \mu_x)^2}{2 \eta^2}}$ is a solution of (\ref{eq:KFsta}). This is the stationary density of $x_i$. The analogous can also be written for $y_i$: $\hat{m}_y(y_i) = \frac{1}{\eta \sqrt{2 \pi}} e^{-\frac{(y_i - \mu_y)^2}{2 \eta^2}}$. \hfill \ensuremath{\blacksquare}

\section{Proof of Proposition 4}\label{ap:proof2}

 The proof follows the the same principle as in chapter \emph{Risk-sensitive mean-field games} in the notes \emph{Mean-field stochastic games} by Tembine. Only the sketch of the proof is given here. To prove the uniqueness of the solution, we suppose that there exists two solutions $(v_{1,t},m_{1,t}),(v_{2,t},m_{2,t})$ of the above system. We want to find a sufficient condition under which the quantity $\int_{\underline{s}} (v_{2,t}(\underline{s})-v_{1,t}(\underline{s}))(m_{2,t}(\underline{s})-m_{1,t}(\underline{s}))\text{d}\underline{s}$ is monotone in time, which is not possible.
\begin{equation}
 \left\{
\begin{array}{cc}
 m_{1,T}(\underline{s}) = m_{2,T}(\underline{s}), \\
v_{1,T'}(\underline{s}) = v_{2,T'}(\underline{s}),
\end{array}
\right.
\end{equation}

Compute the time derivative
\begin{equation}
\begin{aligned}
S_t =& \frac{d}{dt}\biggl(\int_{\underline{s}} (v_{2,t}(\underline{s})-v_{1,t}(\underline{s}))(m_{2,t}(\underline{s})-m_{1,t}(\underline{s}))\text{d}\underline{s}\biggr),\\
=& \int_{\underline{s}} (\frac{\partial{v_{2,t}}}{\partial t} (\underline{s})-\frac{\partial{v_{1,t}}}{\partial t} (\underline{s}))(m_{2,t}(\underline{s})-m_{1,t}(\underline{s}))\text{d}\underline{s},\\
&+ \int_{\underline{s}} (v_{2,t}(\underline{s})-v_{1,t}(\underline{s}))(\frac{\partial{m_{2,t}}}{\partial t}(\underline{s})-\frac{\partial{m_{1,t}}}{\partial t}(\underline{s}))\text{d}\underline{s}.
\end{aligned}
\end{equation}
To express the first term, the difference between the two HJBF equations is taken and multiplied by $m_{2,t}-m_{1,t}$:
\begin{equation}
 \begin{aligned}
&\int_{\underline{s}} (\frac{\partial{v_{2,t}}}{\partial t}-\frac{\partial{v_{1,t}}}{\partial t})(m_{2,t}-m_{1,t})\text{d}\underline{s} =\\
&\int_{\underline{s}}\tilde{H}(\underline{s}(t),\frac{\partial v_{1,t}}{\partial E},m_{1,t})(m_{2,t} - m_{1,t}) \text{d}\underline{s} - \int_{\underline{s}}\tilde{H}(\underline{s}(t),\frac{\partial v_{2,t}}{\partial E},m_{2,t})(m_{2,t} - m_{1,t})\text{d}\underline{s} \\
&+ \int_{\underline{s}}\frac{\eta^2}{2} \frac{\partial^2 v_{1,t}}{\partial h^2}(m_{2,t} - m_{1,t}) \text{d}\underline{s} - \int_{\underline{s}}\frac{\eta^2}{2} \frac{\partial^2 v_{2,t}}{\partial h^2}(m_{2,t} - m_{1,t})\text{d}\underline{s} \\
&+ \frac12 \int_{\underline{s}} \langle\underline{\mu}-\underline{h}(t), \nabla_{\underline{h}}v_{1,t} \rangle_{\mathbb{R}^2}(m_{2,t} - m_{1,t}) \text{d}\underline{s} - \frac12 \int_{\underline{s}} \langle\underline{\mu}-\underline{h}(t), \nabla_{\underline{h}}v_{2,t} \rangle_{\mathbb{R}^2}(m_{2,t} - m_{1,t}) \text{d}\underline{s}.
\end{aligned}
\end{equation}
For the second term, the difference between the two FPK equations is taken and multiplied by $v_{2,t} - v_{1,t}$:
\begin{equation}
\begin{aligned}
&\int_{\underline{s}} (\frac{\partial{m_{2,t}}}{\partial t} - \frac{\partial{m_{1,t}}}{\partial t})(v_{2,t}-v_{1,t})\text{d}\underline{s} =\\
&-\int_{\underline{s}}\frac{\partial}{\partial E}(m_{2,t}\frac{\partial}{\partial u'} \tilde{H}(\underline{s}(t),\frac{\partial v_{2,t}}{\partial E},m_{2,t}))(v_{2,t}-v_{1,t})\text{d}\underline{s}\\
&+\int_{\underline{s}} \frac{\partial}{\partial E}(m_{1,t}\frac{\partial}{\partial u'} \tilde{H}(\underline{s}(t),\frac{\partial v_{1,t}}{\partial E},m_{1,t}))(v_{2,t}-v_{1,t})\text{d}\underline{s}\\
&+\int_{\underline{s}} \frac{\eta^2}{2} \frac{\partial^2 m_{2,t}}{\partial h^2} (v_{2,t}-v_{1,t})\text{d}\underline{s}-\int_{\underline{s}} \frac{\eta^2}{2} \frac{\partial^2 m_{1,t}}{\partial h^2}(v_{2,t}-v_{1,t})\text{d}\underline{s} \\
&+\frac12 \int_{\underline{s}}  div_{\underline{h}}(m_{2,t}(\underline{\mu}-\underline{h}(t))) (v_{2,t}-v_{1,t}) \text{d}\underline{s} -\frac12 \int_{\underline{s}}  div_{\underline{h}}(m_{1,t}(\underline{\mu}-\underline{h}(t))) (v_{2,t}-v_{1,t}) \text{d}\underline{s} .
\end{aligned}
\end{equation}
By integration by parts
\begin{equation}
\int_{\underline{s}} div_{\underline{s}}(k) \phi\text{d}\underline{s} = - \int_{\underline{s}} k\, div_{\underline{s}}(\phi)\text{d}\underline{s},
\end{equation}
then
\begin{equation}
\begin{aligned}
&\int_{\underline{s}} (\frac{\partial{m_{2,t}}}{\partial t} - \frac{\partial{m_{1,t}}}{\partial t})(v_{2,t}-v_{1,t})\text{d}\underline{s} = \\
&\int_{\underline{s}} (m_{2,t}\frac{\partial}{\partial u'} \tilde{H}(\underline{s}(t),\frac{\partial v_{2,t}}{\partial E},m_{2,t}))(\frac{\partial v_{2,t}}{\partial E}-\frac{\partial v_{1,t}}{\partial E})\text{d}\underline{s}\\
&-\int_{\underline{s}} (m_{1,t}\frac{\partial}{\partial u'} \tilde{H}(\underline{s}(t),\frac{\partial v_{1,t}}{\partial E},m_{1,t}))(\frac{\partial v_{2,t}}{\partial E}-\frac{\partial v_{1,t}}{\partial E})\text{d}\underline{s}\\
&+\int_{\underline{s}} \frac{\eta^2}{2} \frac{\partial^2 m_{2,t}}{\partial h^2} (v_{2,t}-v_{1,t})\text{d}\underline{s}-\int_{\underline{s}} \frac{\eta^2}{2} \frac{\partial^2 m_{1,t}}{\partial h^2}(v_{2,t}-v_{1,t})\text{d}\underline{s} \\
&+\frac12 \int_{\underline{s}} m_{2,t} \langle\underline{\mu}-\underline{h}(t), \nabla_{\underline{h}}(v_{2,t}-v_{1,t}) \rangle_{\mathbb{R}^2} \text{d}\underline{s} -\frac12 \int_{\underline{s}} m_{1,t} \langle\underline{\mu}-\underline{h}(t), \nabla_{\underline{h}}(v_{2,t}-v_{1,t}) \rangle_{\mathbb{R}^2} \text{d}\underline{s}.
\end{aligned}
\end{equation}
The full derivative writes
\begin{equation}
 \begin{aligned}
S_t &= \int_{\underline{s}}\tilde{H}(\underline{s}(t),\frac{\partial v_{1,t}}{\partial E},m_{1,t})(m_{2,t} - m_{1,t}) \text{d}\underline{s}-\int_{\underline{s}} \tilde{H}(\underline{s}(t),\frac{\partial v_{2,t}}{\partial E},m_{2,t})(m_{2,t} - m_{1,t}) \text{d}\underline{s} \\
&-\int_{\underline{s}} m_{1,t}\frac{\partial}{\partial u'} \tilde{H}({\underline{s}}(t),\frac{\partial v_{1,t}}{\partial E},m_{1,t})(\frac{\partial v_{2,t}}{\partial E}-\frac{\partial v_{1,t}}{\partial E})\text{d}\underline{s}\\
&+\int_{\underline{s}} m_{2,t}\frac{\partial}{\partial u'} \tilde{H}({\underline{s}}(t),\frac{\partial v_{2,t}}{\partial E},m_{2,t})(\frac{\partial v_{2,t}}{\partial E}-\frac{\partial v_{1,t}}{\partial E})\text{d}\underline{s}. \\
\end{aligned}
\end{equation}
We now introduce
\begin{equation}
m_{\lambda,t} = (1-\lambda)m_{1,t}+\lambda m_{2,t}=m_{1,t} + \lambda(m_{2,t}-m_{1,t}),
\end{equation}
and in the same way
\begin{equation}
v_{\lambda,t} = (1-\lambda)v_{1,t}+\lambda v_{2,t}.
\end{equation}
We study the auxiliary integral
\begin{equation}
\begin{aligned}
C_{\lambda} &= \int_{\underline{s}} \tilde{H}({\underline{s}}(t),\frac{\partial v_{1,t}}{\partial E},m_{1,t})(m_{\lambda,t} - m_{1,t}) \text{d}\underline{s} -\int_{\underline{s}} \tilde{H}({\underline{s}}(t),\frac{\partial v_{\lambda,t}}{\partial E},m_{\lambda,t})(m_{\lambda,t} - m_{1,t}) \text{d}\underline{s} \\
&-\int_{\underline{s}} m_{1,t}\frac{\partial}{\partial u'} \tilde{H}({\underline{s}}(t),\frac{\partial v_{1,t}}{\partial E},m_{1,t})(\frac{\partial v_{\lambda,t}}{\partial E}-\frac{\partial v_{1,t}}{\partial E})\text{d}\underline{s}\\
&+\int_{\underline{s}}m_{\lambda,t}\frac{\partial}{\partial u'}
 \tilde{H}({\underline{s}}(t),\frac{\partial v_{\lambda,t}}{\partial E},m_{\lambda,t})(\frac{\partial v_{\lambda,t}}{\partial E}-\frac{\partial v_{1,t}}{\partial E})\text{d}\underline{s}, \\
\end{aligned}
\end{equation}
which derivative is:
\begin{equation}
\begin{aligned}
\frac{d}{d\lambda}\biggl(\frac{C_{\lambda}}{\lambda}\biggr) &=
- \int_{\underline{s}}\frac{\partial}{\partial m} \tilde{H}({\underline{s}}(t),\frac{\partial v_{\lambda,t}}{\partial E},m_{\lambda,t})(m_{2,t} - m_{1,t})^2 \text{d}\underline{s}\\
&+\int_{\underline{s}} m_{\lambda,t} \frac{\partial^2}{\partial u'^2} \tilde{H}({\underline{s}}(t),\frac{\partial v_{\lambda,t}}{\partial E},m_{\lambda,t})(\frac{\partial v_{2,t}}{\partial E}-\frac{\partial v_{1,t}}{\partial E})^2\text{d}\underline{s} \\
&+\int_{\underline{s}} m_{\lambda,t} \frac{\partial^2}{\partial m \partial u'} \tilde{H}({\underline{s}}(t),\frac{\partial v_{\lambda,t}}{\partial E},m_{\lambda,t})(m_{2,t}-m_{1,t})(\frac{\partial v_{2,t}}{\partial E}-\frac{\partial v_{1,t}}{\partial E})\text{d}\underline{s}.
\end{aligned}
\end{equation}
Note that $\frac{\partial}{\partial m} \tilde{H}(.)$ and $ \frac{\partial^2}{\partial m \partial u'} \tilde{H}(.)$ are functional derivatives. They are defined such that for all $m'\in \mathbb{P}(\mathbb{R}^3)$
\begin{equation}
\begin{aligned}
 &\langle \frac{\partial}{\partial m} \tilde{H}({\underline{s}}(t),\frac{\partial v_{\lambda,t}}{\partial E},m_{\lambda,t}), m'-m_{\lambda,t} \rangle_{\mathbb{P}(\mathbb{R}^3)} = \\
&\lim_{t\to 0} \frac{\tilde{H}({\underline{s}}(t),\frac{\partial v_{\lambda,t}}{\partial E},m_{\lambda,t}+t(m'-m_{\lambda,t})) - \tilde{H}({\underline{s}}(t),\frac{\partial v_{\lambda,t}}{\partial E},m_{\lambda,t})}{t},
\end{aligned}
\end{equation}
and
\begin{equation}
\begin{aligned}
 &\langle \frac{\partial^2}{\partial m \partial u'} \tilde{H}({\underline{s}}(t),\frac{\partial v_{\lambda,t}}{\partial E},m_{\lambda,t}), m'-m_{\lambda,t} \rangle_{\mathbb{P}(\mathbb{R}^3)} =\\
&\lim_{t\to 0} \frac{\frac{\partial \tilde{H}}{\partial u'}({\underline{s}}(t),\frac{\partial v_{\lambda,t}}{\partial E},m_{\lambda,t}+t(m'-m_{\lambda,t})) - \frac{\partial \tilde{H}}{\partial u'}({\underline{s}}(t),\frac{\partial v_{\lambda,t}}{\partial E},m_{\lambda,t})}{t}.
\end{aligned}
\end{equation}
A sufficient condition for the uniqueness of the solution to the mean-field response problem is the monotonicity of the operator associated to
\[
 \left( \begin{array}{cc}
A_{11} & A_{12} \\
A_{21} & a_{22}  \end{array} \right)
\]
with
\begin{equation}
\begin{aligned}
A_{11} =& -\frac{1}{m}\frac{\partial}{\partial m} \tilde{H}, \\
A_{12} =& A_{21}=\frac12 \frac{\partial^2}{\partial m \partial u'} \tilde{H}, \\
a_{22} =& \frac{\partial^2}{\partial u'^2} \tilde{H}.
\end{aligned}
\end{equation}
This is true if
\begin{equation}
 \begin{aligned}
  &A_{11} \succ 0, \quad a_{22} > 0, \\
&A_{12} - \frac{A_{12}^2}{a_{22}} \succ 0.
 \end{aligned}
\end{equation}
\hfill \ensuremath{\blacksquare}


\bibliographystyle{spbasic}      
\bibliography{biblio-book-2011-03-13-bis}

\begin{thebibliography}{33}
\providecommand{\natexlab}[1]{#1}
\providecommand{\url}[1]{{#1}}
\providecommand{\urlprefix}{URL }
\expandafter\ifx\csname urlstyle\endcsname\relax
  \providecommand{\doi}[1]{DOI~\discretionary{}{}{}#1}\else
  \providecommand{\doi}{DOI~\discretionary{}{}{}\begingroup
  \urlstyle{rm}\Url}\fi
\providecommand{\eprint}[2][]{\url{#2}}

\bibitem[{Agarwal and Honig(2012)}]{agarwal-tit-2012}
Agarwal M, Honig M (2012) Adaptive training for correlated fading channels with
  feedback. IEEE Transactions on Information Theory 58(8):5398--5417

\bibitem[{Basar and Olsder(1999)}]{tamer99}
Basar T, Olsder GJ (1999) Dynamic noncooperative game theory. 2nd edition,
  Classics in Applied Mathematics, SIAM, Philadelphia

\bibitem[{Belmega and Lasaulce(2011)}]{belmega-tsp-2011}
Belmega EV, Lasaulce S (2011) Energy-efficient precoding for multiple-antenna
  terminals. IEEE Trans on Signal Processing 59(1):329--340

\bibitem[{Belmega et~al(2009)Belmega, Lasaulce, and Debbah}]{belmega-twc-2009}
Belmega EV, Lasaulce S, Debbah M (2009) Power allocation games for {MIMO}
  multiple access channels with coordination. IEEE Trans on Wireless
  Communications 8(5):3182--3192

\bibitem[{Bonneau et~al(2008)Bonneau, Debbah, Altman, and
  Hj{\o}rungnes}]{bonneau-jsac-2008}
Bonneau N, Debbah M, Altman E, Hj{\o}rungnes A (2008) Non-atomic games for
  multi-user system. IEEE Journal on Selected Areas in Communications
  26(7):1047--1058

\bibitem[{Bressan(2010)}]{Bressan10noncooperativedifferential}
Bressan A (2010) Noncooperative differential games. {A} tutorial

\bibitem[{Buzzi and Saturnino(2011)}]{buzzi-jstp-2011}
Buzzi S, Saturnino D (2011) A game-theoretic approach to energy-efficient power
  control and receiver design in cognitive {CDMA} wireless networks. Journal of
  Selected Topics in Signal Processing 5(1):137--150

\bibitem[{Cover and Thomas(1991)}]{Cover-Book-91}
Cover TM, Thomas JA (1991) Elements of information theory. Wiley-Interscience

\bibitem[{Dumont et~al(2010)Dumont, Hachem, Lasaulce, Loubaton, , and
  Najim}]{dumont-it-2010}
Dumont J, Hachem W, Lasaulce S, Loubaton P, , Najim J (2010) On the capacity
  achieving covariance matrix of rician mimo channels: an asymptotic approach.
  IEEE Trans on Info Theory 56(3):1048--1069

\bibitem[{Evans(2010)}]{evans2010partial}
Evans L (2010) Partial differential equations. American Mathematical Society

\bibitem[{Fette(2006)}]{fette-book-2006}
Fette BA (2006) Cognitive Radio Technology. Elsevier

\bibitem[{Fleming and Soner(1993)}]{fleming1993controlled}
Fleming W, Soner H (1993) Controlled Markov processes and viscosity solutions.
  Applications of mathematics, Springer-Verlag

\bibitem[{Foschini and Miljanic(1993)}]{Foschini-tvt-1993}
Foschini GJ, Miljanic Z (1993) A simple distributed autonomous power control
  algorithm and its convergence. IEEE Trans on Vehicular Technology
  42(4):641--646

\bibitem[{Goodman and Mandayam(2000)}]{goodman-pc-2000}
Goodman DJ, Mandayam NB (2000) Power control for wireless data. IEEE Personal
  Communications 7:48--54

\bibitem[{Gupta and Kumar(1997)}]{gupta-cdc-1997}
Gupta P, Kumar PR (1997) A system and traffic dependent adaptive routing
  algorithm for ad hoc networks. In: IEEE Conf. on Decision and Control (CDC),
  San Diego, USA, pp 2375--2380

\bibitem[{Karatzas and Shreve(1991)}]{Karatzas_Shreve_1991}
Karatzas I, Shreve S (1991) Brownian motion and stochastic calculus. Springer

\bibitem[{Lasaulce and Tembine(2011)}]{lasaulce-book-2011}
Lasaulce S, Tembine H (2011) Game Theory and Learning for Wireless Networks~:
  Fundamentals and Applications. Academic Press

\bibitem[{Lasaulce et~al(2009)Lasaulce, Hayel, Azouzi, and
  Debbah}]{lasaulce-twc-2009}
Lasaulce S, Hayel Y, Azouzi RE, Debbah M (2009) Introducing hierarchy in energy
  games. IEEE Trans on Wireless Comm 8(7):3833--3843

\bibitem[{Lasry and Lions(2007)}]{lasry-jjm-2007}
Lasry JM, Lions PL (2007) Mean field games. Jpn J Math 2(1):229--260

\bibitem[{{Le Treust} and Lasaulce(2010)}]{LeTreustLasaulce(PowerControlRG)10}
{Le Treust} M, Lasaulce S (2010) A repeated game formulation of
  energy-efficient decentralized power control. IEEE Trans on Wireless
  Communications

\bibitem[{M\'{e}riaux et~al(2011{\natexlab{a}})M\'{e}riaux, Hayel, Lasaulce,
  and A.Garnaev}]{meriaux-dsp-2011-b}
M\'{e}riaux F, Hayel Y, Lasaulce S, AGarnaev (2011{\natexlab{a}}) Long-term
  energy constraints and power control in cognitive radio networks. In: IEEE
  Proc. of the 17th International Conference on Digital Signal Processing
  (DSP), Corfu, Greece

\bibitem[{M\'{e}riaux et~al(2011{\natexlab{b}})M\'{e}riaux, Treust, Lasaulce,
  and Kieffer}]{meriaux-dsp-2011-a}
M\'{e}riaux F, Treust ML, Lasaulce S, Kieffer M (2011{\natexlab{b}})
  Energy-efficient power control strategies for stochastic games. In: IEEE
  Proc. of the 17th International Conference on Digital Signal Processing
  (DSP), Corfu, Greece

\bibitem[{Meshkati et~al(2005)Meshkati, Poor, Schwartz, and
  Narayan}]{meshkati-tcom-2005}
Meshkati F, Poor HV, Schwartz SC, Narayan BM (2005) An energy-efficient
  approach to power control and receiver design in wireless data networks. IEEE
  Trans on Communications 53:1885--1894

\bibitem[{Meshkati et~al(2006)Meshkati, Chiang, Poor, and
  Schwartz}]{meshkati-jsac-2006}
Meshkati F, Chiang M, Poor HV, Schwartz SC (2006) A game-theoretic approach to
  energy-efficient power control in multi-carrier {CDMA} systems. IEEE Journal
  on Selected Areas in Communications 24(6):1115--1129

\bibitem[{Mitola and Maguire(1999)}]{mitola-1999}
Mitola J, Maguire GQ (1999) {Cognitive radio: making software radios more
  personal}. IEEE Personal Communications 6(4):13--18

\bibitem[{Olama et~al(2006)Olama, Djouadi, and Charalambous}]{olama-jasp-2006}
Olama M, Djouadi S, Charalambous C (2006) Stochastic power control for
  time-varying long-term fading wireless networks. EURASIP Journal on Applied
  Signal Processing (JASP) Vol. 2006:1--13

\bibitem[{Rodriguez(2003)}]{rodriguez-globecom-2003}
Rodriguez V (2003) An analytical foundation for resource management in wireless
  communication. In: IEEE Proc. of Globecom, pp 898--902

\bibitem[{Saraydar et~al(2002)Saraydar, Mandayam, and
  Goodman}]{saraydar-com-2002}
Saraydar CU, Mandayam NB, Goodman DJ (2002) Efficient power control via pricing
  in wireless data networks. IEEE Trans on Communications 50(2):291--303

\bibitem[{Tembine and Huang(2011)}]{cdc2011}
Tembine H, Huang M (2011) Mean field difference games: {M}c{K}ean-{V}lasov
  dynamics. CDC-ECC, 50th IEEE Conference on Decision and Control and European
  Control Conference

\bibitem[{Tembine et~al(2010)Tembine, Lasaulce, and
  Jungers}]{tembine-crowncom-2010}
Tembine H, Lasaulce S, Jungers M (2010) Joint power control-allocation for
  green cognitive wireless networks using mean field theory. In: IEEE Proc. of
  the 5th Intl. Conf. on Cogntitive Radio Oriented Wireless Networks and
  Communications (CROWNCOM), Cannes, France

\bibitem[{Tse and Hanly(1999)}]{tse-tit-1999}
Tse D, Hanly S (1999) Linear multiuser receivers: Effective interference,
  effective bandwidth and user capacity. In: IEEE Transactions on Information
  Theory, vol~45, pp 641--657

\bibitem[{Tulino and Verd\'{u}(2004)}]{tulino-book-04}
Tulino A, Verd\'{u} S (2004) Random Matrices and Wireless Communications.
  Foundations and trends in communication and information theory, NOW, The
  Essence of Knowledge

\bibitem[{Yates(1995)}]{yates-jsac-1995}
Yates RD (1995) A framework for uplink power control in cellular radio systems.
  IEEE Journal on Selected Areas in Communications 13(7):1341--1347

\end{thebibliography}

\end{document}